\documentclass[12pt]{amsart}
 \pdfoutput=1

\usepackage{amsmath}
\usepackage{amsthm}
\usepackage{latexsym}
\usepackage{amsfonts}
\usepackage{amssymb}
\usepackage{mathtools}
\usepackage{color}
\usepackage{bbm,dsfont}
\usepackage{graphicx}
\usepackage{ifthen}
\usepackage{enumerate}

\newtheorem{proposition}{Proposition}
\newtheorem{theorem}{Theorem}
\newtheorem{lemma}{Lemma}

\newtheorem{remark}{Remark}
\newtheorem{definition}{Definition}

\newcommand{\R}{\mathbb{R}} 
\newcommand{\C}{\mathbb{C}} 
\newcommand{\F}{\mathbb{F}} 

\newcommand{\Z}{\mathbb Z} 

\newcommand{\Tr}[1]{{\rm Tr}\, #1} 

\newcommand{\cc}{\mathcal{C}}

\newcommand{\qq}{\mathcal{Q}}

\newcommand{\lf}{\mathfrak{l}}

\newcommand{\SL}{{\rm SL}(V)}

\newcommand{\hh}{\mathcal{H}} 
\newcommand{\lh}{\mathcal{L(H)}} 
\newcommand{\sym}[2]{S\left(#1 , #2\right)} 
\newcommand{\trq}[1]{{\rm tr}\left[#1\right]} 
\newcommand{\trt}[1]{{\rm tr}\left(#1\right)} 
\newcommand{\ran}{\textrm{ran}\,} 
\newcommand{\id}{\mathbbm{1}} 

\newcommand{\lam}{\lambda}
\newcommand{\eps}{\varepsilon}

\newcommand{\vd}{\mathbf{d}} 
\newcommand{\ve}{\mathbf{e}} 
\newcommand{\vf}{\mathbf{f}} 
\newcommand{\vu}{\mathbf{u}} 
\newcommand{\vv}{\mathbf{v}} 
\newcommand{\vw}{\mathbf{w}} 
\newcommand{\vnull}{\mathbf{0}}

\newcommand{\Po}{\mathsf{Q}}

\renewcommand{\gg}{\mathcal{G}}
\renewcommand{\ss}{\mathcal{D}}

\newcommand{\Aff}[1]{L(#1)} 
\newcommand{\Af}[2]{L_{#1}(#2)} 

\begin{document}

\markboth{C.~Carmeli, J.~Schultz \& A.~Toigo}
{Covariant MUBs in even prime-power dimensional Hilbert spaces}

\title{MAXIMALLY SYMMETRIC STABILIZER MUBS IN EVEN PRIME-POWER DIMENSIONS}

\author{CLAUDIO CARMELI}
\address{\textbf{Claudio Carmeli}; D.I.M.E., Universit\`a di Genova, Via Magliotto 2, Savona, I-17100, Italy}
\email{claudio.carmeli@gmail.com}

\author{JUSSI SCHULTZ}
\address{\textbf{Jussi Schultz}; Turku Centre for Quantum Physics, Department of Physics and Astronomy, University of Turku, Turku, FI-20014, Finland}
\email{jussi.schultz@gmail.com}

\author{ALESSANDRO TOIGO}
\address{\textbf{Alessandro Toigo}; Dipartimento di Matematica, Politecnico di Milano, Piazza Leonardo da Vinci 32, Milano, I-20133, Italy, and I.N.F.N., Sezione di Milano, Via Celoria 16\\
Milano, I-20133, Italy}
\email{alessandro.toigo@polimi.it}

\begin{abstract}
One way to construct a maximal set of mutually unbiased bases (MUBs) in a prime-power dimensional Hilbert space is by means of finite phase-space methods. MUBs obtained in this way are covariant with respect to some subgroup of the group of all affine symplectic phase-space transformations. However, this construction is not canonical: as a consequence, many different choices of  covariance sugroups are possible. In particular, when the Hilbert space is $2^n$ dimensional, it is known that covariance with respect to the full group of affine symplectic phase-space transformations can never be achieved. Here we show that in this case there exist two essentially different choices of maximal subgroups admitting covariant MUBs. For both of them, we explicitly construct a family of $2^n$ covariant MUBs. We thus prove that, contrary to the odd dimensional case,  maximally covariant MUBs  are very far from being unique.
\end{abstract}

\maketitle

\section{Introduction}

The phase-space approach to finite-dimensional quantum mechanics is a very powerful tool in describing quantum systems with finitely many degrees of freedom, and as such it has found numerous applications in quantum tomography and quantum information theory \cite{Wootters87,CCSc90,Leo96,Vourdas04,Gr06,ApBeCh08,Fe11}. This approach works when the Hilbert space of the system is $\hh = \ell^2(\F)$, where $\F$ is any Galois field, and it  employs the analogy of $\hh$ with the Hilbert space $L^2(\R)$ of a free quantum particle moving along the real line. The similarity is carried over by defining a finite dimensional counterpart of the usual Wigner map, and then using it to establish a correspondence between states on $\hh$ and functions on the finite phase-space $\Omega = \F^2$.

The Wigner map is only one instance of the many objects that can be adapted from the infinite dimensional setting by simply turning the real numbers $\R$ into a finite field $\F$ with $q$ elements. Other examples of this correspondence are the finite Heisenberg group and its Schr\"odinger representation on $\hh$ \cite{Schwinger60,AusTol79,Var95}, as well as the finite symplectic group and the associated metaplectic representation \cite{Ho73,Ge77,BaIt86,Ne02,Appleby05JMP,Vourdas05}. The construction we are primarily interested in  is the one that replaces the quadrature observables on $L^2(\R)$ with a set of $q+1$ complementary orthonormal bases on $\hh$. Since such bases  constitute a set of $q+1$ mutually unbiased bases (MUBs), the phase-space approach provides a method for constructing a maximal set of MUBs in the $q$-dimensional Hilbert space $\hh$ \cite{BaBoRoVa02,GiHoWo04,Ho05,SuTo07,DuEnBeYc10}.

Maximal sets of MUBs constructed on the model of quadrature observables are sometimes referred to as {\em stabilizer MUBs} in order to point out their special nature among the family of all maximal MUBs in $\hh$. Their associated orthogonal projections are in a one-to-one correspondence with the set of the affine lines of $\Omega$, in such a way that: (1) all lines parallel to a given direction correspond to projections onto a fixed basis; (2) two sets of parallel lines with different directions correspond to projections onto different bases. Since there are $q+1$ directions in $\Omega$, and $q$ parallel lines for each direction, all the $q(q+1)$ basis vectors are thus achieved.

Being an affine space over $\F$, the finite phase-space $\Omega$ carries the action of the associated group  of  translations $V$; this action clearly descends to the set of the affine lines of $\Omega$, and hence to the corresponding stabilizer MUBs described in the previous paragraph. On the other hand, the group $V$ is represented on $\hh$ by means of the Schr\"odinger representation (usually called {\em Pauli} or {\em Weyl-Heisenberg group} in finite dimensions). Then, by their very definition, stabilizer MUBs are {\em covariant} with respect to such a representation.

However, many possible unitarily inequivalent stabilizer MUBs can be defined over the same phase-space $\Omega$. The source of this ambiguity relies entirely on the fact that one has quite many degrees of freedom in the choice of the correspondence between the lines of $\Omega$ and the bases in the MUBs. It has been shown in~\cite{CaScTo16} that each equivalence class of stabilizer MUBs can be identified by means of a suitable multiplier of $V$, called a {\em Weyl multiplier}, which is uniquely {determined by} the class at hand. One can thus access all the relevant information about some given stabilizer MUBs by simply looking at the properties of their associated Weyl multiplier. This is a single function on $V\times V$ compared to the $q(q+1)$ vectors of the MUBs.

One further property usually required from stabilizer MUBs is covariance with respect to additional symmetries of $\Omega$ other than the phase-space translations. This comes from the fact that, being an affine symplectic space, the phase-space $\Omega$ also carries an action of the symplectic group ${\rm SL}(2,\F)$ and its subgroups. Not all stabilizer MUBs are covariant with respect to such an extended action, but only some very restricted classes. In particular, if the field $\F$ has even characteristic, stabilizer MUBs that are covariant with respect to the full group ${\rm SL}(2,\F)$ {\em do not exist at all} \cite{CaScTo16,Zhu15tris}.

However, covariance with respect to certain subgroups of ${\rm SL}(2,\F)$ is often a very important requirement, which is at the basis of many recent applications to quantum error-correcting codes \cite{DeDM03,GrRoBe03,Sc04}, secure quantum key distributions \cite{Chau05}, entropic uncertainty relations \cite{SuWo07,Su07}, MUB-balanced states \cite{ApBeDa15}, sharply covariant MUBs \cite{Zhu15,Zhu15quater} and unitary designs \cite{Zhu15bis,Zhu15quinquies}. Hence, in the even characteristic case, it is natural to look for all possible subgroups of ${\rm SL}(2,\F)$ admitting covariant stabilizer MUBs.

In this paper, we solve this problem, and show that  maximal  covariance subgroups are divided into {\em two} disjoint conjugacy classes, which are the finite analogues of the {\em maximal split} and {\em maximal nonsplit} toruses of ${\rm SL}(2,\R)$. As in the real case, these two kinds of groups have essentially different actions on the affine lines of $\Omega$, and, correspondingly, on their respective covariant stabilizer MUBs. Indeed, while a split torus permutes the lines preserving two fixed directions, a nonsplit one cycles all the directions,  acting freely on them. On the MUB side, this means that only maximal nonsplit toruses have a transitive action on the set of bases, and thus are the most feasible groups for applications.

The paper is organized as follows. In Section \ref{sec:recall} we recall the essential facts about finite phase-spaces, covariant MUBs and the relation between stabilizer MUBs and Weyl multipliers. In Section \ref{sec:main}, we review the classification of all subgroups of ${\rm SL}(2,\F)$ given in~\cite{Moore04,Wiman899,Dickson58}, and search among them {for those} admitting covariant stabilizer MUBs in even characteristic. Section \ref{sec:cyclic} gives an explicit picture of such subgroups, and it shows that they are either the split or nonsplit toruses in ${\rm SL}(2,\F)$. {The} paper concludes in Section \ref{sec:expl} providing an explicit construction of some maximally covariant stabilizer MUBs in even characteristic. {More precisely, we describe a family of $q$ inequivalent such MUBs, thus proving in particular} that maximally covariant MUBs are not unique. This points out a basic difference with the odd characteristic case, where a unique equivalence class of maximally covariant stabilizer MUBs  is known to exist.

\section{Covariant quadrature systems and Weyl multipliers}\label{sec:recall}

The present section is a brief exposition of the main facts of~\cite{CaScTo16} that will be needed in the following. We refer to Lang's book \cite{LanAlg} for further details on finite fields and Galois theory.

Throughout the paper, $\F$ is a finite field with characteristic $p$. This implies that $|\F|=p^n$ for some positive integer $n$, where we denote by $|\cdot|$ the cardinality of a set. Moreover, $\F$ is an $n$-dimensional vector space over its cyclic subfield $\Z_p$. In this section, the characteristic $p$ may be either even or odd. However, our main results in Sections \ref{sec:main}--\ref{sec:expl} will focus on the case $p=2$.

The {\em trace} of $\F$ is the $\Z_p$-linear functional ${\rm Tr}:\F\to\Z_p$ with $\Tr{\alpha} = \sum_{k=0}^{n-1} \alpha^{p^k}$. We let $\omega_p$ be any $p$-root of unity in the complex field $\C$, and assume $\omega_p$ to be fixed throughout the paper. Note that $\omega_p^{\Tr{\alpha}}$ is a well defined quantity for all $\alpha\in\F$, and exactly $p-1$ possible choices are available for $\omega_p$.

\subsection{Finite phase-space}

In the following, the couple $(\Omega,V)$ {is always} a $2$-dimensional affine space over the field $\F$, that is,
\begin{enumerate}[-]
\item $V$ is a $2$-dimensional vector space over $\F$;
\item $\Omega$ is a set carrying an action of the additive abelian group $V$;
\item the action of $V$ on $\Omega$ is free and transitive.
\end{enumerate}
The translate of an element $x\in\Omega$ by means of a vector $\vv\in V$ is denoted by $x+\vv$. Clearly, $|\Omega|=|V|=|\F|^2$.

We let $\ss$ be the {\em directions} of $\Omega$, that is, the set of $1$-dimensional subspaces of $V$
\begin{equation*}
\ss = \{ D\subset V \mid D = \{\alpha\vd\mid\alpha\in\F\} \text{ for some nonzero $\vd\in V$} \} \, .
\end{equation*}
If $x\in\Omega$, the {\em affine line} (or simply {\em line}) passing through $x$ and parallel to the direction $D\in\ss$ is the subset $x+D = \{x+\vd \mid \vd\in D\}$. There are $|\ss| = |\F|+1$ directions in $\Omega$, hence $|\F|+1$ different lines passing through $x$. Moreover, for a fixed direction $D\in\ss$ there are $|\F|$ disjoint lines parallel to $D$, which form a partition $\Af{D}{\Omega}$ of $\Omega$. The set $\Aff{\Omega} = \bigcup_{D\in\ss} \Af{D}{\Omega}$ is the collection of all the lines of $\Omega$; its cardinality is $|\Aff{\Omega}| = |\F|(|\F|+1)$.

\subsection{Quadrature systems}

Suppose $\hh$ is a finite dimensional Hilbert space with prime-power dimension $\dim\hh = p^n$. A standard way to describe maximal sets of $p^n + 1$ MUBs in $\hh$ is to take the field $\F$ with $|\F| \equiv p^n$ elements, and label each vector of the maximal MUBs with a line of $\Omega$, in such a way that
\begin{enumerate}[-]
\item the $|\F|$ vectors in the same basis correspond to lines parallel to a fixed direction;
\item different bases of the $|\F|+1$ MUBs correspond to different directions.
\end{enumerate}
Changing the labelings of the same MUBs clearly amounts to permuting the bases and the vectors within them. We remark that, in our approach, we regard MUBs with different labelings as {\em essentially distinct}. Anyway, we will not take care of irrelevant phase factors occurring in the vectors of the bases. For this purpose, the most convenient definition of MUBs is in terms of their associated rank-$1$ orthogonal projections as follows.

\begin{definition}\label{def:quad}
A {\em quadrature system} (or simply {\em quadratures}) for the $2$-dimensional affine space $(\Omega,V)$ over $\F$ and acting on the $|\F|$-dimensional Hilbert space $\hh$ is a map $\Po:\Aff{\Omega}\to\lh$, where $\lh$ is the set of the linear operators on $\hh$, such that
\begin{enumerate}[(i)]
\item $\Po(\lf)$ is a rank-$1$ orthogonal projection for all $\lf\in\Aff{\Omega}$;
\item for all $D\in\ss$,
\begin{equation*}
\sum_{\lf\in\Af{D}{\Omega}} \Po(\lf) = \id \,,
\end{equation*}
$\id\in\lh$ being the identity operator of $\hh$;
\item for all $D_1,D_2\in\ss$ with $D_1\neq D_2$,
\begin{equation*}
\trq{\Po(\lf_1)\Po(\lf_2)} = \frac{1}{|\F|} \qquad \text{if $\lf_1 \in \Af{D_1}{\Omega}$ and $\lf_2 \in \Af{D_2}{\Omega}$} \,,
\end{equation*}
where $\trq{\cdot}$ denotes the Hilbert space trace.
\end{enumerate}
\end{definition}

If $\Po$ is a quadrature system for the affine space $(\Omega,V)$, its restriction $\Po_D = \left.\Po\right|_{\Af{D}{\Omega}}$ is a spectral map projecting onto an orthogonal basis, and the spectral maps $\Po_{D_1}$ and $\Po_{D_2}$ project onto two MUBs if $D_1\neq D_2$. A quadrature system thus associates the $|\F|+1$ directions of $\Omega$ with a maximal set of MUBs in $\hh$.

We will regard two unitarily conjugate quadrature systems as essentially the same object. That is, if $\Po_1$ and $\Po_2$ are two quadratures for the same affine space $(\Omega,V)$, acting on possibly different Hilbert spaces $\hh_1$ and $\hh_2$, we say that $\Po_1$ and $\Po_2$ are {\em equivalent} if there is a unitary operator $U:\hh_1\to\hh_2$ such that
\begin{equation}\label{eq:def_equiv}
\Po_2(\lf) = U\Po_1(\lf)U^* \qquad \forall\lf\in\Aff{\Omega} \,.
\end{equation}

\subsection{Symmetries}

The natural symmetry group of the affine space $(\Omega,V)$ is the {\em affine group} ${\rm GL}(V)\rtimes V$, which is the semidirect product of the group ${\rm GL}(V)$ of all invertible $\F$-linear maps of $V$ with the translation group $V$ itself (where $V$ is the normal factor). The action of ${\rm GL}(V)\rtimes V$ on $\Omega$ is the extension of the action of $V$ by translations; it depends on the choice of an {\em origin} $o\in\Omega$, and, once $o$ is fixed, it is given by
\begin{equation*}
(A,\vv) \cdot x = o + A(\vu_{o,x}+\vv) \qquad \forall x\in\Omega,\, (A,\vv)\in {\rm GL}(V)\rtimes V \,,
\end{equation*}
where \(\vu_{o,x}\) is the unique vector such that \(x=o+\vu_{o,x}\).
By means of this formula, we can also define an action of ${\rm GL}(V)\rtimes V$ on $\Aff{\Omega}$, that is,
$$
(A,\vv) \cdot (x+D) = (A,\vv)\cdot x + AD \qquad \forall x+D\in\Aff{\Omega},\, (A,\vv)\in {\rm GL}(V)\rtimes V \,.
$$
{\em Covariance} of a quadrature system is then understood with respect to the latter group action.

\begin{definition}\label{def:VcovarP}
Let $G\subseteq {\rm GL}(V)\rtimes V$ be any subgroup. A quadrature system $\Po$ for the affine space $(\Omega,V)$ acting on the Hilbert space $\hh$ is {\em $G$-covariant} if there exists a unitary projective representation $U$ of $G$ on $\hh$ such that
\begin{equation}\label{eq:defU}
\Po(g\cdot\lf) = U(g)\Po(\lf)U(g)^* \qquad \forall\lf\in\Aff{\Omega},\, g\in G \,.
\end{equation}
\end{definition}

The choice of the unitary operator $U(g)$ in \eqref{eq:defU} is unique up to multiplication by an arbitrary phase factor depending on $g$ (see~\cite[Proposition 3.3]{CaScTo16}); this explains the necessity of dealing with projective representations. We denote by $\qq_G(\Omega)$ the set of all $G$-covariant quadrature systems for the affine space $(\Omega,V)$. If $G\equiv V$ is the group of  phase space translations, a $V$-covariant quadrature system projects on a set of stabilizer MUBs {(or states, codes)} in the terminology of~\cite{DeDM03,GrRoBe03,Sc04,Zhu15,Zhu15quater,Zhu15bis}. Quite many different covariant quadrature systems are then known to exist in this case \cite{GiHoWo04,CaScTo16}. The essential point is that, enlarging the covariance group $G$ to  include elements of ${\rm GL}(V)$, it may happen that the set $\qq_G(\Omega)$ becomes empty.

It is known that in characteristic $p\neq 2$ there is a unique maximal subgruop $G_0\subseteq {\rm GL}(V)$ making the set $\qq_{G_0\rtimes V}(\Omega)$ nonempty, that is, the group $G_0=\SL$ of unit determinant elements in ${\rm GL}(V)$ (see~\cite[Appendix B]{GiHoWo04}). In characteristic $p=2$, however, we have $\qq_{\SL\rtimes V}(\Omega) = \emptyset$ by~\cite[Theorem 7.5]{CaScTo16}, and the problem of finding all the subgroups $G_0\subset {\rm GL}(V)$ admitting $(G_0\rtimes V)$-covariant quadrature systems is open up to now. The objective of the present paper is to solve this question, and thus completely determine the set
\begin{equation}\label{eq:defgg}
\gg = \{G_0\subset {\rm GL}(V) \mid \text{$G_0$ is a subgroup and $\qq_{G_0\rtimes V}(\Omega)\neq\emptyset$}\}
\end{equation}
in even characteristic. Note that also in this case any $G_0\in\gg$ must be a subgroup of $\SL$ by~\cite[Proposition 7.1]{CaScTo16}. Moreover, the set $\gg$ is nontrivial, since by Theorem 8.4 of the same reference the {\em nonsplit toruses} of $\SL$ are elements of $\gg$. The contribution of the present paper is  to show that nonsplit toruses actually do not exhaust the set $\gg$, but they are just `one half' of it.

\subsection{$V$-covariant quadratures and Weyl multipliers}

Our approach to the problem of determining the set $\gg$ relies on the classification of $V$-covariant quadrature systems by means of suitably defined associated multipliers, a topic that was extensively exposed in~\cite{CaScTo16}. Here we briefly recall the essential points.

\begin{theorem}{\cite[Propositions 4.2 and 4.6]{CaScTo16}}\label{teo:Vcov1}
Suppose $\Po$ is a $V$-covariant quadrature system for the affine space $(\Omega,V)$ acting on the Hilbert space $\hh$. Let $o\in\Omega$ be any point. Then there exists a unique projective unitary representation $W_o$ of $V$ on $\hh$ such that
\begin{enumerate}[(W.1)]
\item $W_o(\vv)\Po(\lf)W_o(\vv)^* = \Po(\lf+\vv)$ for all $\lf\in\Aff{\Omega}$ and $\vv\in V$;
\item $W_o(\vd)\Po(o+D) = \Po(o+D)$ for all $D\in\ss$ and $\vd\in D$.
\end{enumerate}
The multiplier $m$ of the projective representation $W_o$ does not depend on the choice of the point $o$, and it satisfies the two relations
\begin{enumerate}[(M.1)]
\item for any $D\in\ss$, $m(\vd_1,\vd_2) = 1$ for all $\vd_1,\vd_2\in D$;
\item $\overline{m(\vu,\vv)} m(\vv,\vu) = \omega_p^{\Tr{\sym{\vu}{\vv}}}$ for all $\vu,\vv\in V$, where $S$ is a symplectic form\footnote{A {\em symplectic form} is a nonzero $\F$-bilinear map \(S:V\times V\to\F\) such that $\sym{\vv}{\vv} = 0$ for all $\vv\in V$. Such a form is symmetric in characteristic $p = 2$ and antisymmetric in characteristic $p\neq 2$. Moreover, since $\F$ is $2$-dimensional, $S$ is uniquely determined up to multiplication by a scalar in $\F_*=\F\setminus\{0\}$.} on $V$ which is uniquely determined.
\end{enumerate}
\end{theorem}

We recall that the {\em multiplier} of $W_o$ is the function $m:V\times V\to \{z\in\C\mid z\overline{z}=1\}$ such that
$$
W_o(\vu+\vv) = m(\vu,\vv)W_o(\vu)W_o(\vv) \qquad \forall \vu,\vv\in V \,.
$$
It satisfies the cocycle relation
$$
m(\vu+\vv,\vw) m(\vu,\vv) = m(\vu,\vv+\vw) m(\vv,\vw) \qquad \forall \vu,\vv,\vw\in V \,.
$$
By items (M.1) and (M.2), the projective representation $W_o$ has the following two additional properties:
\begin{enumerate}[(W.1)]\setcounter{enumi}{2}
\item the restriction $\left.W_o\right|_D$ is an ordinary (i.e., nonprojective) representation of $D$ for all $D\in\ss$;
\item $W_o$ satisfies the commutation relation
$$
W_o(\vu)W_o(\vv) = \omega_p^{\Tr{\sym{\vu}{\vv}}} W_o(\vv)W_o(\vu) \qquad \forall\vu,\vv\in V \,.
$$
\end{enumerate}

A projective unitary representation of $V$ with properties (W.3) and (W.4) is called a {\em Weyl system} for the symplectic space $(V,S)$. The Weyl system $W_o$ satisfying {the additional conditions (W.1) and (W.2) is then said to be} {\em associated} with the $V$-covariant quadratures $\Po$ and {\em centered} at $o$. Accordingly, any multiplier $m$ of the additive abelian group $V$ which satisfies items (M.1) and (M.2) of Theorem \ref{teo:Vcov1} is called a {\em Weyl multiplier} for the symplectic space $(V,S)$. Theorem \ref{teo:Vcov1} then asserts that, through any associated centered Weyl system, an element $\Po\in\qq_V(\Omega)$ defines a symplectic form $S$ on $V$ and a Weyl multiplier $m$ for $(V,S)$ in an unambiguous way. We call such $S$ and $m$ the symplectic form and Weyl multiplier {\em associated} with $\Po$. It is easy to check that, if $\Po_1,\Po_2\in\qq_V(\Omega)$ are equivalent in the sense of \eqref{eq:def_equiv}, then the symplectic forms and Weyl multipliers associated with $\Po_1$ and $\Po_2$ are the same. Remarkably, the converse of this fact also holds.

\begin{theorem}{\cite[Theorem 6.3]{CaScTo16}}\label{teo:Vcov2}
Let $S$ be a symplectic form on $V$, and $m$ a Weyl multiplier for $(V,S)$. Then there exists a unique equivalence class $\qq_V(\Omega,S,m)$ of $V$-covariant quadrature systems for $(\Omega,V)$ having $S$ and $m$ as {the} associated form and multiplier.
\end{theorem}

For any symplectic form $S$ on $V$, Weyl multipliers $m$ for $(V,S)$ exist by~\cite[Proposition 6.1]{CaScTo16} (see also Section \ref{sec:expl} below for some explicit constructions of $m$). Hence the set $\qq_V(\Omega)\supset\qq_V(\Omega,S,m)$ is always nonempty. But what really matters in the equivalence class of a $V$-covariant quadrature system is its Weyl multiplier and not its associated symplectic form. In fact, it is easy to see that for the same symplectic space $(V,S)$ there exist quite many different Weyl multipliers. In other words, if we write $\qq_V(\Omega,S)$ for the totality of $V$-covariant quadrature systems having $S$ as the associated symplectic form, in the chain of inclusions $\qq_V(\Omega)\supset\qq_V(\Omega,S)\supset\qq_V(\Omega,S,m)$ only the latter set is made of a single equivalence class of quadratures.

\subsection{The explicit form of $V$-covariant quadratures}\label{subsec:explVcov}
We assume that the quadrature system $\Po\in\qq_V(\Omega,S,m)$ is given and acts on the Hilbert space $\hh$. In order to write down $\Po$ explicitly, we need to choose
\begin{enumerate}[-]
\item an origin $o\in\Omega$;
\item a {\em symplectic basis} of $(V,S)$, i.e., an $\F$-linear basis $\{\ve_1,\ve_2\}$ of $V$ such that $S(\ve_1,\ve_2) = 1$;
\item a unit vector $\phi_0\in\hh$ in the range of $\Po(o+\F\ve_2)$, where $\F\ve_2 = \{\alpha\ve_2 \mid \alpha\in\F\}$ is the direction in $V$ along $\ve_2$.
\end{enumerate}
After this preparation, if $W_o$ is the Weyl system associated with $\Po$ and centered at $o$, we set
$$
\phi_\gamma = W_o(\gamma\ve_1)\phi_0 \qquad \forall \gamma\in\F\,.
$$
We then have $\phi_\gamma\in\ran[W_o(\gamma\ve_1)\Po(o+\F\ve_2)] = \ran[\Po(o+\gamma\ve_1+\F\ve_2)]$ by covariance of $\Po$. Since
$\Af{\F\ve_2}{\Omega} = \{o+\gamma\ve_1+\F\ve_2 \mid \gamma\in\F\}$, by properties (i) and (ii) of a quadrature system, the vectors $\{\phi_\gamma\mid\gamma\in\F\}$ form an orthonormal basis of $\hh$. In this basis, the Weyl system $W_o$ is given by
\begin{equation}\label{eq:WHexpl}
W_o(\alpha_1\ve_1+\alpha_2\ve_2)\phi_\gamma = m(\alpha_1\ve_1,\alpha_2\ve_2) \omega_p^{-\Tr{\alpha_2\gamma}} \phi_{\gamma+\alpha_1} \qquad \forall \alpha_1,\alpha_2\in\F \,.
\end{equation}
Indeed,
\begin{align*}
& W_o(\alpha_1\ve_1+\alpha_2\ve_2)\phi_\gamma = m(\alpha_1\ve_1,\alpha_2\ve_2) W_o(\alpha_1\ve_1)W_o(\alpha_2\ve_2)W_o(\gamma\ve_1) \phi_0 \\
& \qquad \qquad = m(\alpha_1\ve_1,\alpha_2\ve_2) \omega_p^{\Tr{\sym{\alpha_2\ve_2}{\gamma\ve_1}}} W_o(\alpha_1\ve_1)W_o(\gamma\ve_1)W_o(\alpha_2\ve_2) \phi_0\\
& \qquad \qquad = m(\alpha_1\ve_1,\alpha_2\ve_2)\omega_p^{-\Tr{\alpha_2\gamma}} \phi_{\gamma+\alpha_1}
\end{align*}
since $W_o(\alpha_2\ve_2) \phi_0 = \phi_0$ {because $W_o$ is centered at $o$ (see item (W.2) of Theorem \ref{teo:Vcov1}).} By~\cite[Proposition 5.2]{CaScTo16}, for all $\vu,\vv\in V$ with $\vu\neq\vnull$,
$$
\Po(o+\vv+\F\vu) = \frac{1}{|\F|} \sum_{\lam\in\F} \omega_p^{\Tr{\sym{\vv}{\lam\vu}}} W_o(\lam\vu) \,,
$$
where $\F\vu$ is the direction along $\vu$, and hence
\begin{equation}\label{eq:quadexpl}
\begin{aligned}
\Po(o+\vv+\F\vu)\phi_\gamma & = \frac{1}{|\F|} \sum_{\lam\in\F} m(\lam\alpha_1\ve_1,\lam\alpha_2\ve_2) \,\omega_p^{\Tr{\lam\left[\alpha_2(\beta_1 - \gamma) - \alpha_1\beta_2\right]}} \phi_{\gamma+\lam\alpha_1} \\
& \qquad \qquad \text{with} \qquad \vu = \alpha_1\ve_1+\alpha_2\ve_2 \qquad \vv = \beta_1\ve_1+\beta_2\ve_2\,.
\end{aligned}
\end{equation}

In the converse direction, if $m$ is any given Weyl multiplier for the symplectic space $(V,S)$, one can pick a $|\F|$-dimensional Hilbert space $\hh$, fix an orthonormal basis $\{\phi_\gamma\mid\gamma\in\F\}$ of $\hh$, and define the maps $W_o$ and $\Po$ as in \eqref{eq:WHexpl} and \eqref{eq:quadexpl}. As $W_o$ and $\Po$ are unitarily equivalent to the maps defined in the previous paragraph, we have $\Po\in\qq_V(\Omega,S,m)$ and $W_o$ is its associated Weyl system centered at $o$.

\subsection{More symmetries besides translations}

We already noticed that $G_0\subseteq\SL$ is a necessary condition for the set $\qq_{G_0\rtimes V}(\Omega)$ to be nonempty (see~\cite[Proposition 7.1]{CaScTo16}). In order to find a sufficient condition, we need the notion of {\em $G_0$-invariance} for a Weyl multiplier $m$, that is,
$$
m(A\vu,A\vv) = m(\vu,\vv) \qquad \forall \vu,\vv\in V,\, A\in G_0 \,.
$$
The existence of $G_0$-invariant Weyl multipliers is equivalent to the set $\qq_{G_0\rtimes V}(\Omega)$ being nonempty. Indeed, we have the following fact.

\begin{proposition}{\cite[Proposition 7.2]{CaScTo16}} \label{prop:Gquad-->Gmult}
Let $G_0\subseteq\SL$ be any subgroup. A quadrature system $\Po\in\qq_V(\Omega)$ is $(G_0\rtimes V)$-covariant if and only if its associated Weyl multiplier is $G_0$-invariant.
\end{proposition}

As a consequence, the set $\gg$ of \eqref{eq:defgg} coincides with
$$
\gg = \{G_0\subseteq\SL\mid \text{there exist $G_0$-invariant Weyl multipliers}\}\,.
$$

In odd characteristic, the multiplier
\begin{equation}\label{eq:WHmult_odd}
m(\vu,\vv) = \omega_p^{\Tr{S(2^{-1}\vv,\vu)}} \,,
\end{equation}
is a Weyl multiplier for the symplectic space $(V,S)$ which is invariant with respect to the whole group $\SL$ (see~\cite[Proposition 7.4]{CaScTo16}); therefore, $\gg$ is actually the set of all the subgroups of $\SL$. However, in even characteristic such an $m$ can not be defined, and we need to look for subgroups $G_0\subset\SL$ admitting $G_0$-invariant Weyl multipliers case by case. This is done in the next section, and the detailed description of the set $\gg$ in characteristic $p=2$ is provided in Section \ref{sec:cyclic} below (see Theorem \ref{theo:main}).

\section{All covariant quadrature systems in characteristic $2$}\label{sec:main}

From now on, we focus on characteristic $p=2$. The following is then the key step towards our characterization of the set $\gg$ in this case.

\begin{lemma}\label{lem:main}
Suppose $\F$ has characteristic $p=2$. Then $\qq_{G_0\rtimes V}(\Omega) = \emptyset$ for all subgroups $G_0\subseteq\SL$ such that $|G_0|$ is even.
\end{lemma}

Before proving the lemma, observe that in characteristic $p=2$ we have $+1 = -1$ in $\F$, and $\omega_2 = -1$ is the unique possible choice of a $2$-root of unity in $\C$. Moreover, the square map $\alpha\mapsto\alpha^2$ is an automorphism of $\F$ over $\Z_2$. Its inverse is the map $\alpha\mapsto\alpha^{1/2} = \alpha^{|\F|/2}$.

\begin{proof}[Proof of Lemma \ref{lem:main}]
By Proposition \ref{prop:Gquad-->Gmult}, it is enough to show that, if $G_0$ has even order, there do not exist $G_0$-invariant Weyl multipliers. So, let us assume by contradiction that $|G_0|$ is even and $m$ is a $G_0$-invariant Weyl multiplier. By Cauchy theorem (see~\cite[p.~97]{Suzuki82}), there exists an order $2$ element in $G_0$, that is, a symplectic map $A\in G_0$ such that $A\neq I$ and $A^2 = I$. Let $\ve_2\in V$ be such that $A\ve_2\neq\ve_2$. Then $\ve_1 = A\ve_2 + \ve_2 \neq 0$ because $+1=-1$, and $A\ve_1 = \ve_1$. Hence the vectors $\{\ve_1,\ve_2\}$ are linearly independent, and thus form an $\F$-linear basis of $V$. In particular, $S(\ve_1,\ve_2) = \alpha \neq 0$ since $S\neq 0$. Possibly rescaling $\ve_2$ by $\alpha^{-1}$, we can assume that $\{\ve_1,\ve_2\}$ is a symplectic basis of $(V,S)$. The conditons $\det A = 1$ and $A\ve_1 = \ve_1$ imply that $A$ is upper triangular with diagonal entries $(1,1)$ in the basis $\{\ve_1,\ve_2\}$; that is, $A\ve_2 = \beta\ve_1 + \ve_2$ for some $\beta\neq 0$.

Now, choose $\gamma\in\F$ such that ${\rm Tr}\,\gamma = 1$ (this is always possible by~\cite[Theorem VI.5.2]{LanAlg}). Let
$$
\vf_1 = (\beta\gamma)^{1/2} \, \ve_1 \qquad \qquad \vf_2 = (\beta^{-1}\gamma)^{1/2} \, \ve_2 \,.
$$
Then
\begin{equation}\label{eq:1}
A\vf_1 = \vf_1 \qquad \qquad A\vf_2 = \vf_1+\vf_2 
\end{equation}
and
\begin{equation}\label{eq:2}
{\rm Tr}\, S(\vf_1,\vf_2) = 1 \,.
\end{equation}
{We have
\begin{align*}
1 & = m(\vf_1+\vf_2,\vf_1+\vf_2) \qquad\qquad \text{(property (M.1))}\\
& = m(\vf_1+\vf_2,\vf_1+\vf_2) m(\vf_1,\vf_2) \overline{m(\vf_1,\vf_2)} \\
& = m(\vf_1+\vf_2+\vf_1,\vf_2) m(\vf_1+\vf_2,\vf_1) \overline{m(\vf_1,\vf_2)} \qquad\qquad \text{(multiplier property)} \\
& = m(\vf_2,\vf_2) m(A\vf_2,A\vf_1) \overline{m(\vf_1,\vf_2)} \qquad\qquad\text{(by \eqref{eq:1})}\\
& = m(\vf_2,\vf_1) \overline{m(\vf_1,\vf_2)} \qquad\qquad \text{(property (M.1) and $G_0$-invariance)} \\
& = (-1)^{{\rm Tr}\, \sym{\vf_1}{\vf_2}} \qquad\qquad \text{(property (M.2))} \\
& = -1 \qquad\qquad \text{(by \eqref{eq:2})}\,,
\end{align*}
which is the desired contradiction.}
\end{proof}

The next step is to list all the possible subgroups of $\SL$. By the previous result, for $p=2$ all the subgroups having even order can be dropped from $\gg$. The classification of the subgroups of the finite projective unimodular group ${\rm PSL}(V) = \SL /\{I,-I\}$ goes back to Moore and Wiman's papers \cite{Moore04,Wiman899}, which cover both the even and odd characteristic case (see~\cite[pp.~285-286]{Dickson58} for a summary of the subgroups found by Moore and Wiman). Note that ${\rm PSL}(V) = \SL$ {for $p=2$}, hence in our case~\cite{Moore04,Wiman899,Dickson58} actually enumerate all the subgroups of $\SL$. For the present purposes,  we use here the more modern version of Moore and Wiman's classification given in Suzuki's book \cite{Suzuki82}.

\begin{theorem}\label{teo:Suzuki}
In characteristic $p=2$, any subgroup of $\SL$ is isomorphic to one of the following groups.
\begin{enumerate}[(a)]
\item The dihedral groups of order $2(|\F|\pm 1)$ and their subgroups.
\item A group $H$ of order $|\F|(|\F|-1)$ and its subgroups. A Sylow $2$-subgroup $Q$ of $H$ is isomorphic to $\Z_2^k$, $Q$ is normal in $H$, and the factor group $H/Q$ is a cyclic group of order $|\F|-1$.
\item The alternating groups $A_4$ or $A_5$.
\item {${\rm SL}(V')$, where $V'$ is a $2$-dimensional vector space over a subfield $\F'\subseteq\F$.}
\end{enumerate}
\end{theorem}
\begin{proof}
This is an immediate application of~\cite[Theorems III.(6.25) and III.(6.26)]{Suzuki82}, when $q=|\F|$ is even, since ${\rm PSL}(V) = \SL$ in this case. In particular, each item follows from the corresponding one in Suzuki's Theorem III.(6.25) by observing that
\begin{itemize}
\item[(a,b)] the greatest common divisor of $2$ and $|\F|-1$ is $d=1$, and using~\cite[I.(9.14)]{Suzuki82}, for the characterization of the elementary abelian $2$-groups defined in II.(5.22) therein;
\item[(c)] $\SL$ has no subgroups isomorphic to the symmetric group $\Sigma_4$ by~\cite[item (iii) of Theorem III.(6.26)]{Suzuki82};
\item[(d)] if $\F'$ is any field such that $|\F'|^m = |\F|$, then ${\rm PGL}(2,\F') = {\rm PSL}(2,\F') = {\rm SL}(2,\F')$ since $\F'$ is a subfield of $\F$ and hence also has even characteristic.
\end{itemize}
\end{proof}

Combining Lemma \ref{lem:main} and Theorem \ref{teo:Suzuki} we obtain the following conclusion.

\begin{proposition}\label{prop:iffodd}
Let $p=2$, and suppose $S$ is {any} symplectic form on $V$. Then the set $\qq_{G_0\rtimes V}(\Omega,S) = \qq_{G_0\rtimes V}(\Omega)\cap\qq_V(\Omega,S)$ is not empty if and only if $G_0$ is a cyclic group with $|G_0|$ odd.
\end{proposition}
\begin{proof}
The proof of sufficiency is a straightforward adaptation of the proof of~\cite[Proposition 8.3]{CaScTo16}. Indeed, suppose $G_0$ is a cyclic group with odd order. If $m_0$ is any Weyl multiplier for the symplectic space $(V,S)$, let $m(\vu,\vv) = \prod_{A\in G_0} m_0(A\vu,A\vv)$. Then $m$ is a multiplier of $V$, which clearly satisfies $\left. m\right|_{D\times D} = 1$ for all $D\in\ss$, since all its factors do it. Since $\overline{m_0(A\vu,A\vv)} m_0(A\vv,A\vu) = (-1)^{\Tr{\sym{A\vu}{A\vv}}} = (-1)^{\Tr{\sym{\vu}{\vv}}}$ for every $A\in G_0$, we have
$$
\overline{m(\vu,\vv)} m(\vv,\vu) = (-1)^{|G_0|\Tr{\sym{\vu}{\vv}}} = (-1)^{\Tr{\sym{\vu}{\vv}}}
$$
because $|G_0|$ is odd. Therefore, $m$ satisfies items (M.1) and (M.2) of Theorem \ref{teo:Vcov1}, that is, it is a Weyl multiplier for $(V,S)$. For all $B\in G_0$,
$$
m(B\vu,B\vv) = \prod_{A\in G_0} m_0(AB\vu,AB\vv) = \prod_{A\in G_0} m_0(A\vu,A\vv) = m(\vu,\vv) \,,
$$
which shows that $m$ is $G_0$-invariant. Hence $\qq_{G_0\rtimes V}(\Omega,S)\supset\qq_V(\Omega,S,m)\neq\emptyset$ by Theorem \ref{teo:Vcov2} and Proposition \ref{prop:Gquad-->Gmult}.

Conversely, if $\qq_{G_0\rtimes V}(\Omega)\neq\emptyset$, then $|G_0|$ is odd by Lemma \ref{lem:main}. So, we need to check which ones of the groups listed in Theorem \ref{teo:Suzuki} have odd order. Since $|A_4| = 12$ and $|A_5| = 60$, the possibilities in item (c) of Theorem \ref{teo:Suzuki} are excluded. Moreover, by~\cite[p.~81]{Suzuki82} we have $|{\rm SL}(V')| = |\F'|(|\F'|^2 - 1)$, which is even when $V'$ is a vector space over a subfield $\F'\subseteq\F$; hence $G_0$ can not be as in item (d) of Theorem \ref{teo:Suzuki}. Thus, items (a) and (b) are the only remaining possibilities.

The dihedral group $D_{2n}$ is the semidirect product $\Z_2\rtimes\Z_n$, where the nontrivial element $1\in\Z_2$ acts on the normal factor $\Z_n$ as
$$
(1,0)(0,x)(1,0)^{-1} = (0,-x) \qquad \forall x\in\Z_n \,.
$$
If $G_0$ is a subgroup of $D_{2(|\F|\pm 1)}$ and $(z,x)\in G_0$, then $z=0$, as otherwise $(1,x)^2 = (0,-x+x) = (0,0)$ implying that $|G_0|$ is even. It follows that $G_0\subseteq\Z_{|\F|\pm 1}$, hence $G_0$ is a cyclic group.

Finally, suppose $G_0\subseteq H$, where $H$ is {as in item} (b) of Theorem \ref{teo:Suzuki}. Then the subgroup $Q_0 = Q\cap G_0$ is normal in $G_0$, and the quotient group $G_0/Q_0$ is naturally identified with a subgroup of $H/Q$. Since $Q$ is a Sylow $2$-subgroup of $H$, either $Q_0$ is trivial or its order is even; hence $Q_0$ is trivial because $|G_0|$ is odd. Since $H/Q$ is cyclic of order $|\F|-1$, also its subgroup $G_0/Q_0 = G_0$ is cyclic.

In conclusion, $|G_0|$ {being} odd implies that $G_0$ is cyclic, and this concludes the proof.
\end{proof}

\section{Cyclic subgroups of $\SL$}\label{sec:cyclic}

By Proposition \ref{prop:iffodd},
$$
{\gg = \{G_0\subset\SL\mid \text{$G_0$ is cyclic and with odd order}\}\quad\text{in characteristic $p=2$}\,.}
$$
We will shortly see that the cyclic subgroups $G_0\subset\SL$ divide into three classes, each class  being determined by the eigenvalues of any of its generators. Recall that the eigenvalues of an arbitrary symplectic map $A\in\SL$ are the roots of its characteristic polynomial
\begin{equation}\label{eq:char_pol}
p_A(X) = \det(A-XI) = X^2 - {\rm tr}(A) X + 1 \,,
\end{equation}
and thus they are two possibly coincident elements $\xi_1$ and $\xi_2$ of the quadratic extension $\tilde{\F}$ of $\F$. Since $p_A$ has coefficients in $\F$, either $\xi_1,\xi_2\in\F$ or $\xi_1,\xi_2\in\tilde{\F}\setminus\F$, and in the latter case $\xi_2 = \overline{\xi_1}$, where $\overline{\xi_1} = \xi_1^{|\F|}$ is the conjugate of $\xi_1$. Both of the eigenvalues are nonzero, and they satisfy the relations $\xi_1 + \xi_2 = \trt{A}$ and $\xi_1 \xi_2 = 1$. In particular, in even characteristic the equality $\xi_1 = \xi_2$ holds if and only if $\xi_1 = \xi_2 = 1$, and in this case $\trt{A} = 0$.

Again, for the remaining of the section we restrict ourselves to even characteristic. The following terminology then summarizes all the possibilities for an element $A\in\SL$ (see e.g.~\cite[p.~95]{Hump75}).
\begin{definition}\label{def:defsubgrSL}
In characteristic $p=2$, an element $A\in\SL$ is
\begin{enumerate}[-]
\item {\em split}, if $A=I$ or $A$ has two different eigenvalues $\xi,\xi^{-1}\in\F$;
\item {\em nonsplit}, if $A$ has two different eigenvalues $\xi,\xi^{-1}\in\tilde{\F}\setminus\F$, with $\xi^{-1} = \overline{\xi}$;
\item {\em unipotent}, if $A\neq I$ and $1$ is the sole eigenvalue of $A$.
\end{enumerate}
$A$ is {\em semisimple} if it is either split or nonsplit.
\end{definition}

Let us fix a basis of $V$ over $\F$, and write any element $A\in\SL$ as a unit determinant $2\times 2$ matrix with entries in $\F$ with respect to such a basis. If $A\in\SL$ is semisimple and $\xi,\xi^{-1}\in\tilde{\F}$ are its two eigenvalues, then
\begin{equation*}
A = U \left(\begin{array}{cc} \xi & 0 \\ 0 & \xi^{-1} \end{array}\right) U^{-1} \qquad \text{for some $2\times 2$ matrix $U$ with entries in $\tilde{\F}$}\,.
\end{equation*}
All the entries of $U$ can be chosen in $\F$ if and only if $A$ is split. In any case, $A^k = I$ if and only if $\xi^k=\xi^{-k}=1$, that is, the order $k_0$ of $A$ and $\xi$ coincide. Hence,
\begin{enumerate}[-]
\item if $A$ is split, then $k_0$ divides the order of the cyclic multiplicative group $\F_*$ of the nonzero elements of $\F$, which is $|\F_*| = |\F|-1$;
\item if $A$ is nonsplit, then $k_0$ divides the order of the cyclic group $M = \{\xi\in\tilde{\F}_*\mid \xi\overline{\xi} = 1\}$, which is $|\F|+1$ (see~\cite[Section 8]{CaScTo16} for a simple proof).
\end{enumerate}
Finally, for $0<k<k_0$, the eigenvalues of $A^k$ are $\xi^k$ and $\xi^{-k}$. Therefore, if $A$ is semisimple, then also $A^k$ is semisimple for all $0<k<k_0$.

On the other hand, if $A$ is unipotent, there is a nonzero $\ve_1\in V$ such that $A\ve_1 = \ve_1$. To find the order of $A$, pick a vector $\ve_2\in V$ linearly independent from $\ve_1$. Then $A\ve_2 = \alpha\ve_2+\beta\ve_1$, with $\alpha=1$ by the unit determinant condition, and $\beta\neq 0$ because $A\neq I$. Moreover, we have $A^2\ve_1 = \ve_1$ and $A^2\ve_2 = \ve_2 + 2\beta\ve_1 = \ve_2$, hence $A^2 = I$. In particular, the order of $A$ is $2$.

This discussion shows that the next definition is consistent and exhausts all the cyclic subgroups of $\SL$.

\begin{definition}\label{def:toruses}
A cyclic subgroup of $\SL$ is a {\em torus} [respectively, a {\em split torus, nonsplit torus, unipotent subgroup}] if it is generated by a semisimple [resp., split, nonsplit, unipotent] element of $\SL$.
\end{definition}

Definitions \ref{def:defsubgrSL} and \ref{def:toruses} can be easily extended to odd $p$. It is then a general fact, valid in all characteristics, that there exists a maximal split [respectively, nonsplit] torus $T\subset\SL$, and all split [resp., nonsplit] toruses of $\SL$ are conjugate to subgroups of $T$. Moreover, all the unipotent subgroups of $\SL$ are conjugate in even characteristic, and they are divided into four conjugacy classes when $p\neq 2$. Indeed, this follows from~\cite[\S 6]{Moore04} (see also~\cite[pp.~262--268]{Dickson58} and \cite[III.(6.23)]{Suzuki82}). Here we report the following elementary proof in characteristic $p=2$.

\begin{proposition}\label{prop:cyclic}
Suppose $p=2$.
\begin{enumerate}[(a)]
\item There exists a split [respectively, nonsplit] torus $T\subset\SL$ such that $|T| = |\F|-1$ [resp., $|T| = |\F|+1$]. Any split [resp., nonsplit] torus has odd order and is conjugated to a subgroup of $T$. In particular, all toruses of the same order are conjugated.
\item There exists a unique conjugacy class of unipotent subgroups in $\SL$. All unipotent subgroups have order $2$.
\end{enumerate}
\end{proposition}
\begin{proof}
We preliminarly prove that, if $\xi\in\tilde{\F}$ is such that $\xi + \xi^{-1} \in \F$, then the conjugacy class of the symplectic map
\begin{equation}\label{eq:defA}
A_\xi = \left(\begin{array}{cc} \xi + \xi^{-1} & 1 \\ 1 & 0 \end{array}\right)
\end{equation}
is the set
$$
\cc(A_\xi) = \{A\in\SL\setminus\{I\}\mid \text{$\xi$ and $\xi^{-1}$ are the eigenvalues of $A$}\} \,.
$$
Indeed, by \eqref{eq:char_pol} the latter set is $\cc(A_\xi) = \{A\in\SL\setminus\{I\}\mid \trt{A} = \xi + \xi^{-1}\}$. Therefore, $A\in\cc(A_\xi)$, and it suffices to show that any $A\in\SL\setminus\{I\}$ is such that
$$
A = U \left(\begin{array}{cc} \trt{A} & 1 \\ 1 & 0 \end{array}\right) U^{-1} \qquad \text{for some $U\in\SL$} \,.
$$
Writing $A$ in  matrix form
$$
A = \left(\begin{array}{cc} \alpha & \beta \\ \gamma & \delta \end{array}\right) \qquad \text{with} \qquad \alpha,\beta,\gamma,\delta\in\F, \ \alpha\delta+\beta\gamma = 1 \,,
$$
it can be directly verified that a possible choice of $U$ is
$$
U = \begin{cases}
\left(\begin{array}{ccc} 0 &\,& \beta^{1/2} \\ \beta^{-1/2} &\,& \alpha\beta^{-1/2} \end{array}\right) & \text{ if $\beta\neq 0$} \\
\left(\begin{array}{ccc} \gamma^{-1/2} &\,& \delta\gamma^{-1/2} \\ 0 &\,& \gamma^{1/2} \end{array}\right) & \text{ if $\gamma\neq 0$} \\
(1+\alpha)^{-1}\left(\begin{array}{ccc} \alpha & 1 \\ 1 & \alpha \end{array}\right) & \text{ if $\beta = \gamma = 0$ and $\delta=\alpha^{-1}$}
\end{cases}
$$
thus proving the claim.

In order to prove (a), observe first of all that $\F_* = \{1\}$ if and only if $\F = \Z_2$, and the claims for split toruses are trivial in this case since $T=\{I\}$ is the unique split torus of $\SL$. Next, suppose $\xi\neq 1$ is a generator of the cyclic group $\F_*$ [resp., $M = \{\xi\in\tilde{\F}_*\mid \xi\overline{\xi} = 1\}$], and define $A_\xi$ as in \eqref{eq:defA}. Then $A_\xi$ is a split [resp., nonsplit] element of the same order as $\xi$, that is, $A_\xi$ generates a split [resp., nonsplit] torus $T$ of order $|T|=|\F_*|=|\F|-1$ [resp., $|T|=|M|=|\F|+1$]. If $T'$ is any split [resp., nonsplit] torus generated by a split [resp., nonsplit] element $A'\in\SL$, either $A'=I$ or $A'$ has two different eigenvalues $\xi'$ and $\xi^{\prime\,-1}$ with $\xi'\in\F_*$ [resp., $\xi'\in M$]. It follows that $\xi'=\xi^k$ for some $k$, hence $A'\in\cc(A_\xi^k)$ by the previous claim. Therefore, $T'$ is conjugated to the cyclic subgroup of $T$ generated by $A_\xi^k$. Since $T$ has a unique cyclic subgroup of each order dividing $|T|$ (see~\cite[Proposition I.4.2 and I.4.3(iv)]{LanAlg}), all split [resp., nonsplit] toruses of the same order are conjugated among them and with a unique subgroup of $T$. Finally, $|\F|-1 = 2^r - 1$ and $|\F|+1 = 2^r + 1$ are relatively prime, hence two toruses $T_1$ and $T_2$ such that $|T_1| = |T_2|$ are either both split or both nonsplit, and so they are conjugated.

The proof of (b) follows since by definition any unipotent element $B\in\SL$ is such that $B\neq I$ and $\xi = \xi^{-1} = 1$ are the two eigenvalues of $B$; all unipotent $B$'s are then conjugated to
$$
A_1 = \left(\begin{array}{cc}
0 & 1 \\ 1 & 0
\end{array}\right)
$$
by the claim at the beginning of the proof. Since $A_1^2 = I$, the same holds for $B$.
\end{proof}

\begin{remark}
The proof of Proposition \ref{prop:cyclic} also yields an explicit expression for a symplectic map $A$ generating a maximal cyclic subgroup of $\SL$. Indeed, such a map is given by \eqref{eq:defA} with $\xi$ a generator of either $\F_*$ or $M$ in the case of a maximal torus, or $\xi+\xi^{-1} = 0$ for unipotent subgroups. 
\end{remark}

We are now in position to state and prove the main result of the paper.

\begin{theorem}\label{theo:main}
In characteristic $p=2$, the set $\qq_{G_0\rtimes V}(\Omega,S)$ is nonempty if and only if $G_0$ is a torus. The maximal subgroups $G_0\subset\SL$ admitting $(G_0\rtimes V)$-covariant quadrature systems are either maximal split or maximal nonsplit toruses.
\end{theorem}
\begin{proof}
The theorem immediately follows by combining Propositions \ref{prop:iffodd} and \ref{prop:cyclic}.
\end{proof}

\section{Maximally invariant Weyl multipliers}\label{sec:expl}

Up to now, we have considered the existence problem for $(G_0\rtimes V)$-covariant quadrature systems. However, when the set $\qq_{G_0\rtimes V}(\Omega,S)$ is nonempty, we have neither investigated whether it is made up of a unique equivalence class of quadratures, nor have we explicitly written down any of its elements.

In this section, we fill this gap in the case where $G_0\equiv T$ is a maximal torus in even characteristic, providing  many examples of inequivalent elements in $\qq_{T\rtimes V}(\Omega,S)$. Moreover, for all these examples we exhibit a unitary projective representation $U$ of $G=T\rtimes V$ yielding the covariance relation \eqref{eq:defU}.

By Theorem \ref{teo:Vcov2} and Proposition \ref{prop:Gquad-->Gmult}, the equivalence classes of quadratures in the set $\qq_{T\rtimes V}(\Omega,S)$ are in one-to-one correspondence with the $T$-invariant Weyl multipliers for the symplectic space $(V,S)$. If such a multiplier $m$ is given, Section \ref{subsec:explVcov} provides the explicit construction of the corresponding quadrature system in terms of $m$ (see \eqref{eq:quadexpl}). The main difficulty is then to write down an explicit expression for a $T$-invariant Weyl multiplier.

Note that an explicit formula for the multiplier $m$ also allows one to construct the projective representation $U$ of $T$ yielding the $T$-covariance of $\Po$. This follows from the next theorem.
\begin{theorem}{\cite[Theorem 8.5]{CaScTo16}}\label{teo:metap}
In any characteristic, let $T$ be a maximal torus, and suppose $\Po\in\qq_{T\rtimes V} (\Omega)$. Let $W_o$ be the Weyl system associated with $\Po$ and centered at the point $o\in\Omega$ such that ${\rm GL}(V)\cdot o = \{o\}$, and let $m$ be its Weyl multiplier. Then a possible choice for the projective representation $U$ of $T$ appearing in \eqref{eq:defU} is
\begin{equation}\label{eq:metap}
U(A) = \frac{1}{|\F|} \sum_{\vu\in V} m(\vu,(A-I)^{-1}\vu) W_o(\vu) \qquad \forall A\in T\setminus\{I\}\,.
\end{equation}
\end{theorem}
\begin{proof}
If $T$ is nonsplit, this is Theorem 8.5 of~\cite{CaScTo16}. The proof of the latter result uses only the two facts that $A-I$ is invertible on $V$, and $-I  \in T$. These facts are still true if $T$ is split, hence the same proof works without any change also in the split case.
\end{proof}
The operators $W_o(\vu)$ appearing in Theorem \ref{teo:metap} are explicitly given in formula \eqref{eq:WHexpl}, which again only depends on the Weyl multiplier $m$ associated with $\Po$.

For the remaining part of the section, we turn to the problem of characterizing the $T$-invariant Weyl multipliers in characteristic $p=2$. We remark that the present discussion is a refinement of~\cite[Appendix B]{CaScTo16}, which outlines how to find a $T$-invariant Weyl multiplier by averaging a noninvariant one over $T$, but does not contain a compact formula for the result.

First of all, observe that, although in odd characteristic there is the natural choice of the Weyl multiplier \eqref{eq:WHmult_odd}, which takes its values in the set of the $p$-roots of unity and is actually invariant with respect to the whole group $\SL$, when $p=2$ a more elaborate construction is required. The key difference is that in the latter case there is no $\pm 1$-valued Weyl multiplier at all. Indeed, if $m$ were such a multiplier, then, for $\vf_1,\vf_2\in V$ with $\Tr{\sym{\vf_1}{\vf_2}} = 1$, we would get the contradiction
\begin{align*}
1 & = m(\vf_1+\vf_2,\vf_1+\vf_2) = m(\vf_1+\vf_2,\vf_1+\vf_2)m(\vf_1,\vf_2)^2 \\
& = m(\vf_1,\vf_2+\vf_1+\vf_2)m(\vf_2,\vf_1+\vf_2)m(\vf_1,\vf_2) \\
& = m(\vf_1,\vf_1)m(\vf_2,\vf_1+\vf_2)(-1)^{\Tr{\sym{\vf_2}{\vf_1}}}m(\vf_2,\vf_1) \\
& = -m(\vf_2,\vf_2+\vf_1)m(\vf_2,\vf_1) = -m(\vf_2+\vf_2,\vf_1)m(\vf_2,\vf_2)\\
& = -1 \,.
\end{align*}
Actually, in the even characteristic case the minimal possible choice of a Weyl multiplier is
$$
m(\vu,\vv) = i^{g(\vu,\vv)}
$$
for some function $g:V\times V\to\Z_4$. The function $g$ must clearly be a $\Z_4$-valued multiplier. Moreover, properties (M.1) and (M.2) of a Weyl multiplier become
\begin{enumerate}[(M'.1)]
\item for any $D\in\ss$, $g(\vd_1,\vd_2) = 0$ for all $\vd_1,\vd_2\in D$;
\item $g(\vv,\vu) - g(\vu,\vv) = 2\Tr{\sym{\vu}{\vv}}$ for all $\vu,\vv\in V$, where the map $z\mapsto 2z$ goes from $\Z_2$ to $\Z_4$.
\end{enumerate}
The additional condition that $m$ is $T$-invariant then requires that $g(A\vu,A\vv) = g(\vu,\vv)$ for some generator $A$ of $T$ and all $\vu,\vv\in V$.

In order to construct $g$, we need the fact that in even characteristic there exists a linear basis $\{\omega_1 , \omega_2 , \ldots , \omega_n\}$ of $\F$ over $\Z_2$ such that $\Tr{(\omega_i \omega_j)} = \delta_{i,j}$ for all $i,j\in\{1,2,\ldots,n\}$ (see~\cite[Theorem 4]{SeLe80}). After choosing such a basis, we also fix a sequence $r_1,r_2 ,\ldots,r_n$ with $r_i = \pm 1$. Then we define the following map $h : \F\to\Z_4$
$$
h\left(\sum_{i=1}^n z_i\omega_i\right) = \sum_{i=1}^n r_i z_i^2 \qquad \forall z_1,\ldots,z_n\in\Z_2 \,.
$$
Note that $h$ is well defined, since the map $z\mapsto z^2$ is well defined from $\Z_2$ to $\Z_4$. Clearly, $h(0)=0$. Moreover, if $\alpha = \sum_i z_i\omega_i$ and $\beta = \sum_i t_i\omega_i$, then
\begin{equation}\label{eq:prop_h}
h(\alpha+\beta) = h(\alpha)+h(\beta)+2\Tr{\alpha\beta} \,.
\end{equation}
The construction of $g$ is slightly different in the two cases in which the maximal torus $T$ is split or nonsplit.

\subsection{The split case}
Let $A$ be a generator of $T$ with eigenvalues $\xi,\xi^{-1}\in\F$, and let $\{\ve_1,\ve_2\}$ be vectors of $V$ such that $A\ve_1 = \xi\ve_1$ and $A\ve_2 = \xi^{-1} \ve_2$. Possibly rescaling $\ve_2$, we can assume that $\{\ve_1,\ve_2\}$ is a symplectic basis of $(V,S)$. We use this basis to define the following two $\F$-bilinear forms $B_+$ and $B_-$ on $V$
$$
B_+(\vu,\vv) = B_-(\vv,\vu) = \sym{\vu}{\ve_1}\sym{\vv}{\ve_2} \qquad \forall\vu,\vv\in V \,.
$$
Since $B_+(\vu,\vu) = B_-(\vu,\vu)$ for all $\vu\in V$, the sum $B_+ + B_-$ is a symplectic form on $V$. As $(B_+ + B_-)(\ve_1,\ve_2) = 1$, actually
$$
B_+ + B_- = S \,.
$$
Moreover, since $\sym{A\vu}{\ve_1} = \sym{\vu}{A^{-1}\ve_1} = \xi^{-1}\sym{\vu}{\ve_1}$ and similarly $\sym{A\vu}{\ve_2} = \xi\sym{\vu}{\ve_2}$, the bilinear forms $B_+$ and $B_-$ are $T$-invariant, that is
$$
B_+(A\vu,A\vv) = B_+(\vu,\vv) \qquad \text{and} \qquad B_-(A\vu,A\vv) = B_-(\vu,\vv) \qquad \forall\vu,\vv\in V \,.
$$
We then define a $\Z_4$-valued multiplier $g_0$ on $V$, given by
$$
g_0(\vu,\vv) = 2\Tr{B_+(\vu,\vv)} = 2\Tr{B_-(\vv,\vu)} \,.
$$
(That $g_0$ is a multiplier follows from its biadditivity property $g_0(\vu_1+\vu_2,\vv) = g_0(\vu_1,\vv)+g_0(\vu_2,\vv)$ and $g_0(\vu,\vv_1+\vv_2) = g_0(\vu,\vv_1)+g_0(\vu,\vv_2)$.) Condition (M'.2) holds for $g_0$. However, to make also condition (M'.1) satisfied, we need to introduce the equivalent $\Z_4$-valued multiplier $g$, with
\begin{align}\label{eq:defgsplit}
g(\vu,\vv) & = h(B_+(\vu+\vv,\vu+\vv)^{1/2}) - h(B_+(\vu,\vu)^{1/2}) - h(B_+(\vv,\vv)^{1/2}) + g_0(\vu,\vv) \,.
\end{align}
Indeed, for all $\lam,\mu\in\F$, by the property \eqref{eq:prop_h} of $h$,
\begin{eqnarray*}
g(\lam\vu,\mu\vu) & = & h((\lam+\mu)B_+(\vu,\vu)^{1/2}) - h(\lam B_+(\vu,\vu)^{1/2}) - h(\mu B_+(\vu,\vu)^{1/2}) \\
&& + 2\Tr{\lam\mu B_+(\vu,\vu)} \\
& = & 0 \,.
\end{eqnarray*}
Finally, from the analogous property of $B_+$ it immediately follows that $g(A\vu,A\vv) = g(\vu,\vv)$ for all $\vu,\vv\in V$, hence $g$ is $T$-invariant.

We have thus found the $T$-invariant Weyl multiplier $m=i^g$. We can use the construction of Section \ref{subsec:explVcov}, with the symplectic basis $\{\ve_1,\ve_2\}$ given by the above eigenbasis of $A$, in order to exhibit the quadrature system $\Po\in\qq_{T\rtimes V}(\Omega,S)$ having $m$ as its associated multiplier. To this aim, it is enough to evaluate
\begin{equation}\label{eq:mult_split}
m(\alpha_1\ve_1,\alpha_2\ve_2) = i^{h\left((\alpha_1\alpha_2)^{1/2}\right)}
\end{equation}
and insert it into \eqref{eq:WHexpl}, \eqref{eq:quadexpl} to get
\begin{align*}
\Po(o+\vv+\F\vu)\phi_\gamma & = \frac{1}{|\F|} \sum_{\lam\in\F} i^{h\left(\lam \left(\alpha_1\alpha_2\right)^{1/2}\right)} (-1)^{\Tr{\lam\left[\alpha_2(\beta_1  + \gamma) + \alpha_1\beta_2\right]}} \phi_{\gamma+\lam\alpha_1} \\
& \qquad \qquad \text{with} \qquad \vu = \alpha_1\ve_1+\alpha_2\ve_2 \qquad \vv = \beta_1\ve_1+\beta_2\ve_2\,
\end{align*}
with its associated centered Weyl system
$$
W_o(\alpha_1\ve_1+\alpha_2\ve_2)\phi_\gamma = i^{h\left((\alpha_1\alpha_2)^{1/2}\right)} (-1)^{\Tr{\alpha_2\gamma}} \phi_{\gamma+\alpha_1} \,.
$$
In order to determine the unitary operator $U(A)$ yielding the $T$-covariance, we can either use \eqref{eq:metap} or simply notice that $U(A)\phi_0 = c\phi_0$ for some scalar $c\in\C$, since $U(A)\Po(o+\F\ve_2) = \Po(o+\F\ve_2)U(A) \equiv c\Po(o+\F\ve_2)$ by $T$-covariance. On the other basis vectors,
$$
U(A)\phi_\gamma = U(A)W_o(\gamma\ve_1)\phi_0 = W_o(\gamma A\ve_1)U(A)\phi_0 = c W_o(\gamma\xi\ve_1)\phi_0 = c\phi_{\gamma\xi} \,.
$$
$U$ becomes an ordinary representation of $T$ by setting $c=1$.

\subsection{The nonsplit case}\label{subsec:nonsplit}
Let $A$ and $\xi,\xi^{-1}$ be as in the previous case. Now, $\xi,\xi^{-1}\in\tilde{\F}\setminus\F$ with $\xi^{-1} = \overline{\xi}$, and $A$ is diagonalized in the extension $\tilde{V} = \tilde{\F}\otimes_\F V$ of $V$ to the scalars $\tilde{\F}$. Let $\ve\in\tilde{V}$ be a nonzero vector such that $A\ve = \xi\ve$. Then $A\overline{\ve} = \overline{A\ve} = \overline{\xi}\overline{\ve}$, where we still denote by $\overline{\,\cdot\,}$ the $\tilde{\F}$-antilinear map on $\tilde{V}$ which restricts to the identity on $V$. {The $\F$-bilinear form $S$ uniquely extends to a symplectic form on $\tilde{V}$.} Note that $\sym{\overline{\vu}}{\overline{\vv}} = \overline{\sym{\vu}{\vv}}$. In particular, $\sym{\ve}{\overline{\ve}}\in\F$, hence, possibly rescaling both $\ve$ and $\overline{\ve}$ by the factor $\sym{\ve}{\overline{\ve}}^{-1/2}$, we can assume that $\{\ve,\overline{\ve}\}$ is a symplectic basis of $(\tilde{V},S)$. Now, as in the split case we define the $\tilde{\F}$-bilinear forms on $\tilde{V}$
$$
B_+(\vu,\vv) = B_-(\vv,\vu) = \sym{\vu}{\ve}\sym{\vv}{\overline{\ve}} \qquad \forall\vu,\vv\in \tilde{V} \,.
$$
Again, $B_+(\vu,\vu) = B_-(\vu,\vu)$ for all $\vu\in\tilde{V}$,  $B_+ + B_- = S$, and  the forms $B_+$ and $B_-$ are $T$-invariant. Moreover, although \(B_+\) and \(B_-\) are \(\tilde{\F}\)-valued bilinear forms, the corresponding quadratic forms are \(\F\) valued:  $B_+(\vu,\vu) = B_-(\vu,\vu)\in\F$ for all $\vu\in V$. Let $\widetilde{\rm Tr} : \tilde{\F}\to\Z_2$ be any $\Z_2$-linear extension of ${\rm Tr}$ to $\tilde{\F}$. (For example, if $\zeta$ is any element of $\tilde{\F}\setminus\F$, we can set $\widetilde{\rm Tr}(\alpha+\beta\zeta) = \Tr{\alpha}$ for all $\alpha,\beta\in\F$.) We then define the following $\Z_4$-valued biadditive multiplier $g_0$ on $V$
$$
g_0(\vu,\vv) = 2\widetilde{\rm Tr}\, B_+(\vu,\vv) = 2\widetilde{\rm Tr}\, B_-(\vv,\vu) \qquad\forall\vu,\vv\in V
$$
and its equivalent multiplier $g$ as in formula \eqref{eq:defgsplit}. Since $g_0$ satisfies condition (M'.2), so does $g$. Moreover, $g$ also fulfills (M'.1) and is $T$-invariant, the computation being the same as in the split case. In conclusion, $m = i^g$ is a $T$-invariant Weyl multiplier on $V$.

As in the previous section, we are now going to explicitly exhibit the $T\rtimes V$-covariant quadrature system $\Po\in\qq_V(\Omega,S,m)$ along the lines of Section \ref{subsec:explVcov}. In the present case, we fix the following symplectic basis $\{\ve_1,\ve_2\}$ of $V$
\begin{equation}\label{eq:nonsplit_basis}
\ve_1 = \overline{\eps}\ve + \eps\overline{\ve} \qquad \ve_2 = \eps\ve + \overline{\eps}\overline{\ve} \qquad \text{with} \qquad \eps = (\xi+1)^{-1/2} \,.
\end{equation}
Moreover, we choose the extension $\widetilde{\rm Tr}$ such that $\widetilde{\rm Tr}\,\eps^2 = 0$. Then, with some manipulations (reported in Appendix \ref{app:calc}),
\begin{align}
& m(\alpha_1\ve_1,\alpha_2\ve_2) = i^{h\left((\alpha_1\alpha_2)^{1/2}\right)}(-1)^{\Tr{\left[\alpha_1\alpha_2\eps\overline{\eps}+(\alpha_1+\alpha_2)(\alpha_1\alpha_2\eps\overline{\eps})^{1/2}\right]}} \label{eq:mult_non_split} \\
& \begin{aligned}
& m(\alpha_1\ve_1+\alpha_2\ve_2\,,\,(A+I)^{-1}(\alpha_1\ve_1+\alpha_2\ve_2)) = \\
& \qquad = i^{-\left[h\left(\alpha_1 (\eps\overline{\eps})^{1/2}\right) + h\left(\alpha_2 (\eps\overline{\eps})^{1/2}\right) + h\left((\alpha_1\alpha_2)^{1/2}\right)\right]}(-1)^{\Tr{\left[\alpha_1\alpha_2\eps\overline{\eps} + (\alpha_1+\alpha_2)(\alpha_1\alpha_2\eps\overline{\eps})^{1/2}\right]}} \,. 
\end{aligned}\label{eq:mult_non_split_2}
\end{align}
By \eqref{eq:WHexpl}, \eqref{eq:quadexpl},
\begin{align*}
& \Po(o+\vv+\F\vu)\phi_\gamma = \frac{1}{|\F|} \sum_{\lam\in\F} i^{h\left(\lam(\alpha_1\alpha_2)^{1/2}\right)} \\
& \qquad\qquad\qquad \times (-1)^{\Tr{\lam\left\{ \lam\left[\alpha_1\alpha_2\eps\overline{\eps}+(\alpha_1+\alpha_2)(\alpha_1\alpha_2\eps\overline{\eps})^{1/2}\right]+\alpha_2(\beta_1 + \gamma) + \alpha_1\beta_2\right\}}} \phi_{\gamma+\lam\alpha_1} \\
& W_o(\vu)\phi_\gamma = i^{h\left((\alpha_1\alpha_2)^{1/2}\right)}(-1)^{\Tr{\left[\alpha_1\alpha_2\eps\overline{\eps}+(\alpha_1+\alpha_2)(\alpha_1\alpha_2\eps\overline{\eps})^{1/2}+\alpha_2\gamma\right]}} \phi_{\gamma+\alpha_1} \\
& \qquad\qquad\qquad\qquad\qquad\qquad\qquad \text{with} \qquad \vu = \alpha_1\ve_1+\alpha_2\ve_2 \qquad \vv = \beta_1\ve_1+\beta_2\ve_2 \,.
\end{align*}
Moreover, by \eqref{eq:metap},
\begin{align*}
U(A)\phi_\gamma & = \frac{1}{|\F|} \sum_{\alpha_1,\alpha_2\in\F} m(\alpha_1\ve_1+\alpha_2\ve_2\,,\,(A+I)^{-1}(\alpha_1\ve_1+\alpha_2\ve_2)) \\
& \qquad \times W_o(\alpha_1\ve_1+\alpha_2\ve_2)\phi_\gamma \\
& = \frac{1}{|\F|} \sum_{\alpha_1,\alpha_2\in\F} i^{-\left[h\left(\alpha_1(\eps\overline{\eps})^{1/2}\right) + h\left(\alpha_2(\eps\overline{\eps})^{1/2}\right) \right]}(-1)^{\Tr{\alpha_2\gamma}} \phi_{\gamma+\alpha_1} \,.
\end{align*}

As a final consideration, observe that in both the split and nonsplit cases our construction provides a quite big amount of different $T$-invariant Weyl multipliers. Indeed, for a fixed choice of the orthonormal basis $\{\omega_1 , \omega_2 , \ldots , \omega_n\}$ of $\F$ over $\Z_2$, changing the sequence of signs $r_1,r_2 ,\ldots,r_n$ in the definition of $h$ yields $2^n$ different Weyl multipliers $m$; this can be seen by direct inspection of \eqref{eq:mult_split} and \eqref{eq:mult_non_split}. Consequently, the set $\qq_{T\rtimes V}(\Omega,S)$ contains at least $2^n$ inequivalent quadratures. This shows that in even characteristic there exists a large degree of arbitrarity in the choice of a maximally covariant quadrature system.

\section{Conclusions}

We have found all the extended symmetries of stabilizer MUBs in even prime-power dimensions beyond the basic group $V$ of phase-space translations. We have proved that only two inequivalent such extensions are possible, namely by means of either a split or a nonsplit torus $T\subset\SL$. In particular, it turns out that both of the possibilities give rise to whole families of inequivalent maximally symmetric stabilizer MUBs, contrasting with the case in odd prime-power dimensions, where the maximal symmetry requirement points out a single class of stabilizer MUBs. For each of the two extensions, we have focused on a particular family of inequivalent maximally symmetric stabilizer MUBs, providing both the explict form of the MUBs (more precisely, of their associated rank-$1$ projections, that we named {\em quadrature system}) and the expression of the covariance operators.

In the applications, one is usually interested in finding the smallest groups of unitary operators cycling all the bases in a given maximal set of MUBs \cite{Chau05,SuWo07,Su07,Zhu15,Zhu15quater}. For maximally symmetric stabilizer MUBs, this corresponds to requiring a maximal nonsplit torus as the extra symmetry group (see~\cite[Section 8]{CaScTo16}), since split toruses do not cycle the two bases corresponding to the directions they keep fixed.

As a final consideration, in our approach the symmetry properties of stabilizer MUBs are essentially related to their labelings with the phase-space lines. Indeed, for any pair of $V$-covariant quadratures $\Po_1$ and $\Po_2$, the two sets of rank-$1$ projections $\ran\Po_i = \{\Po_1(\lf)\mid\lf\in\Aff{\Omega}\}$ ($i=1,2$) are always unitarily conjugated by~\cite[Theorem 7.9]{CaScTo16}, although of course $\Po_1$ and $\Po_2$ may not be equivalent in the sense of \eqref{eq:def_equiv}. The present paper thus essentially dealt with the problem of how to arrange the phase-space labeling of stabilizer MUBs in order to make them `as much covariant as possible'. As pointed out in~\cite[Remark 7.8]{CaScTo16}, the covariance operators $U(g)$ satisfying \eqref{eq:defU} do not exhaust all unitaries preserving the (unlabeled) set of projections $\ran\Po$ of some $\Po\in\qq_{G_0\rtimes V}(\Omega)$. However, they are the only ones whose action on MUBs can be naturally related to a phase-space structure.

\section*{Acknowledgments}

The authors wish to thank Prof.~Markus Grassl for suggesting the problem. JS acknowledges financial support from the EU through the Collaborative Projects QuProCS (Grant Agreement No. 641277).

\appendix

\section{Supplemental material}\label{app:calc}

Here we provide the explicit calculations leading to \eqref{eq:mult_non_split} and \eqref{eq:mult_non_split_2}. For the $\Z_4$-valued multiplier $g$ found in Section \ref{subsec:nonsplit}, in the basis \eqref{eq:nonsplit_basis} and for $\alpha_1,\alpha_2\in\F$, we have
\begin{align*}
& g(\alpha_1\ve_1,\alpha_2\ve_2) = h\left(\left[\sym{\alpha_1\ve_1+\alpha_2\ve_2}{\ve}\sym{\alpha_1\ve_1+\alpha_2\ve_2}{\overline{\ve}}\right]^{1/2}\right) \\
& \qquad\qquad - h\left(\left[\sym{\alpha_1\ve_1}{\ve}\sym{\alpha_1\ve_1}{\overline{\ve}}\right]^{1/2}\right) - h\left(\left[\sym{\alpha_2\ve_2}{\ve}\sym{\alpha_2\ve_2}{\overline{\ve}}\right]^{1/2}\right) \\
& \qquad\qquad + 2\widetilde{\rm Tr}\, \sym{\alpha_1\ve_1}{\ve}\sym{\alpha_2\ve_2}{\overline{\ve}} \\
& \qquad = h\left(\left[(\alpha_1\eps+\alpha_2\overline{\eps})(\alpha_1\overline{\eps}+\alpha_2\eps)\right]^{1/2}\right) \\
& \qquad\qquad - h\left(\alpha_1\left(\eps\overline{\eps}\right)^{1/2}\right) - h\left(\alpha_2\left(\overline{\eps}\eps\right)^{1/2}\right) + 2\widetilde{\rm Tr}\, \alpha_1\alpha_2\eps^2 \\
& \qquad = h\left(\left[(\alpha_1^2+\alpha_2^2)\eps\overline{\eps}+\alpha_1\alpha_2\right]^{1/2}\right) \qquad\qquad\qquad\quad \text{because $\eps^2+\overline{\eps}^2 = 1$}\\
& \qquad\qquad - h\left(\alpha_1\left(\eps\overline{\eps}\right)^{1/2}\right) - h\left(\alpha_2\left(\overline{\eps}\eps\right)^{1/2}\right) \qquad\qquad \text{because $\widetilde{\rm Tr}\,\eps^2 = 0$} \\
& \qquad = h\left((\alpha_1+\alpha_2)(\eps\overline{\eps})^{1/2}+(\alpha_1\alpha_2)^{1/2}\right) \qquad \qquad \quad \text{by $\Z_2$-linearity of $\cdot\,^{1/2}$}\\
& \qquad\qquad - h\left(\alpha_1\left(\eps\overline{\eps}\right)^{1/2}\right) - h\left(\alpha_2\left(\overline{\eps}\eps\right)^{1/2}\right) \\
& \qquad = h\left((\alpha_1\alpha_2)^{1/2}\right) + 2\Tr{\left[\alpha_1\alpha_2\eps\overline{\eps}+(\alpha_1+\alpha_2)(\alpha_1\alpha_2\eps\overline{\eps})^{1/2}\right]} \qquad \text{by \eqref{eq:prop_h}} \,.
\end{align*}
This proves \eqref{eq:mult_non_split}. Concerning \eqref{eq:mult_non_split_2},
\begin{align*}
& g(\alpha_1\ve_1+\alpha_2\ve_2\,,\,(A+I)^{-1}(\alpha_1\ve_1+\alpha_2\ve_2)) = \\
& \qquad = g(\alpha\ve+\overline{\alpha}\overline{\ve}\,,\,(A+I)^{-1}(\alpha\ve+\overline{\alpha}\overline{\ve})) \qquad\qquad\qquad \text{with $\alpha = \alpha_1\overline{\eps}+\alpha_2\eps$} \\
& \qquad = g(\alpha\ve+\overline{\alpha}\overline{\ve}\,,\,\eps^2\alpha\ve+\overline{\eps}^2\overline{\alpha}\overline{\ve}) \\
& \qquad = h\left(\left[\sym{(1+\eps^2)\alpha\ve+(1+\overline{\eps}^2)\overline{\alpha}\overline{\ve}}{\ve}\sym{(1+\eps^2)\alpha\ve+(1+\overline{\eps}^2)\overline{\alpha}\overline{\ve}}{\overline{\ve}}\right]^{1/2}\right) \\
& \qquad\qquad - h\left(\left[\sym{\alpha\ve+\overline{\alpha}\overline{\ve}}{\ve}\sym{\alpha\ve+\overline{\alpha}\overline{\ve}}{\overline{\ve}}\right]^{1/2}\right) \\
& \qquad\qquad - h\left(\left[\sym{\eps^2\alpha\ve+\overline{\eps}^2\overline{\alpha}\overline{\ve}}{\ve}\sym{\eps^2\alpha\ve+\overline{\eps}^2\overline{\alpha}\overline{\ve}}{\overline{\ve}}\right]^{1/2}\right) \\
& \qquad = h\left(\left[(1+\overline{\eps}^2)\overline{\alpha}(1+\eps^2)\alpha\right]^{1/2}\right) - h\left(\left(\overline{\alpha}\alpha\right)^{1/2}\right) - h\left(\left(\overline{\eps}^2\overline{\alpha}\eps^2\alpha\right)^{1/2}\right) \\
& \qquad = - h\left(\left(\overline{\alpha}\alpha\right)^{1/2}\right) \qquad\qquad\qquad\qquad\qquad\qquad\qquad \text{because $\eps^2+\overline{\eps}^2 = 1$} \\
& \qquad = - h\left(\left[(\alpha_1^2+\alpha_2^2)\eps\overline{\eps} + \alpha_1\alpha_2\right]^{1/2}\right) \qquad\qquad\qquad\quad \text{because $\eps^2+\overline{\eps}^2 = 1$} \\
& \qquad = - h\left((\alpha_1+\alpha_2)(\eps\overline{\eps})^{1/2} + (\alpha_1\alpha_2)^{1/2}\right) \qquad\qquad\quad \text{by $\Z_2$-linearity of $\cdot\,^{1/2}$} \\
& \qquad = - \left[h\left(\alpha_1(\eps\overline{\eps})^{1/2}\right) + h\left(\alpha_2(\eps\overline{\eps})^{1/2}\right) + h\left((\alpha_1\alpha_2)^{1/2}\right)\right] \\
& \qquad\quad + 2\Tr{\left[\alpha_1\alpha_2\eps\overline{\eps} + (\alpha_1+\alpha_2)(\alpha_1\alpha_2\eps\overline{\eps})^{1/2}\right]} \qquad\qquad\qquad\qquad\quad \text{by \eqref{eq:prop_h}}\,,
\end{align*}
which gives \eqref{eq:mult_non_split_2}.

\end{document}